\documentclass[a4paper,fubak]{scrartcl}
\usepackage[utf8]{inputenc}
\usepackage{amsmath,amssymb}
\usepackage[final]{microtype}
\usepackage{mathtools,xspace}
\usepackage{authblk}
\usepackage[final]{hyperref}
\usepackage[capitalise]{cleveref}
\usepackage[draft]{fixme}

\usepackage{caption,subcaption}
\usepackage{tikz}
\usetikzlibrary{trees,shapes.geometric,calc}
\tikzstyle{k}=[draw,circle,inner sep=1pt,minimum width=2pt,fill,font=\small]
\tikzstyle{nu}=[draw,circle,inner sep=2pt,minimum width=10pt,font=\small]
 
\hypersetup{
    colorlinks=false,
    pdfborder={0 0 0},
    }
\usepackage{microtype}
\usepackage{amsthm}
\newtheorem{theorem}{Theorem}[section]
\newtheorem{lemma}{Lemma}[section]
\newtheorem{definition}{Definition}[section]

\newtheorem{corollary}{Corollary}[section]
\newtheorem{fact}{Fact}
\crefname{fact}{Fact}{Facts}
\DeclareMathOperator{\homc}{HOMC}
\DeclareMathOperator{\sharpp}{\#P}

\newcommand{\field}{\ensuremath{K}}

\newcommand{\ie}{i.e.\xspace}

\newcommand{\vnp}{\ensuremath{\text{VNP}}\xspace}
\newcommand{\vp}{\ensuremath{\text{VP}}\xspace}
\newcommand{\vac}{\ensuremath{\text{VAC}_0}\xspace}

\newcommand{\leqc}{\leq_c}
\DeclareMathOperator{\gf}{GF}
\DeclareMathOperator{\F}{\mathcal{F}}

\newcommand{\Fcycle}{\F_\text{cycle}}
\newcommand{\Fclique}{\F_\text{clique}}
\newcommand{\Fouterplanar}{\F_\text{outerplanar}}
\newcommand{\Fplanar}{\F_\text{planar}}

\newcommand{\Ftree}{\F_\text{tree}}

\newcommand{\Fgenus}{\F_{\text{genus}(k)}}

\newcommand{\cycleeven}{\ensuremath{\mathcal{UHC}_{\text{$n_{0}$,even}}}}

\title{Dichotomy Theorems for Homomorphism Polynomials of Graph Classes}
\author{Christian Engels\footnote{Saarland University, Department of Computer Science\\ \texttt{engels@cs.uni-saarland.de}}}

\begin{document}
\maketitle
\begin{abstract}
    In this paper, we will show dichotomy theorems for the computation of polynomials corresponding
    to evaluation of graph homomorphisms in Valiant's model.  We are given a fixed graph $H$ and
    want to find all graphs, from some graph class, homomorphic to this $H$. These graphs will be
    encoded by a family of polynomials.

    We give dichotomies for the polynomials for cycles, cliques, trees, outerplanar graphs, planar
    graphs and graphs of bounded genus.
\end{abstract}

\section{Introduction}
\label{sec:introduction}
Graph homomorphisms are studied because they give important generalizations of many natural
questions ($k$-coloring, acyclicity, binary CSP and many more cf.~\cite{hell2004graphs}). One of the
first results, given by Hell and Ne{\v{s}}et{\v{r}}il \cite{Hell199092}, was on the decision problem
where they gave a dichotomy. The exact result was, that deciding if there exists a homomorphism from
some graph $G$ to a fixed undirected graph $H$ is polynomial time computable if $H$ is
bipartite and NP-complete otherwise. A different side of graph homomorphisms was looked at by
Chekuri and Rajaraman \cite{Chekuri:1997:CQC:645502.656110} Dalmau, Kolaitis and Vardi
\cite{dalmau2002}, and Freuder \cite{Freuder:1990:CKS:1865499.1865500} and finally Grohe
\cite{Grohe:2007:CHC:1206035.1206036}. They studied the following: Given a restricted graph class
$\mathcal{G}$, decide if there is a graph $G\in \mathcal{G}$ homomorphic to a given graph
$H$. Instead of restricting the graph $H$ as in the first problem, we restrict the graph classes we
map from. Later, focus shifted onto the counting versions of these two sides where we have to count
the number of homomorphisms. Dyer and Greenhill \cite{dyer2000complexity} solved the first problem
in the counting case and Dalmau and Jonsson \cite{dalmau2004complexity} the second. The first
problem was extended by Bulatov and Grohe \cite{Bulatov2005148} to graphs with multiple edges. They
also notice some interesting connections to statistical physics and constraint satisfaction
problems. A good introduction to the history of graph homomorphism was written by Grohe and Thurley
\cite{grohe2011counting} and research on these topics continues even today with two noticeable being
the works by Goldberg, Grohe, Jerrum and Thurley \cite{GGJT10} and by Cai, Chen and Lu
\cite{CCL13}.

However, the arithmetic circuit complexity was still open. The previous results could only show that
the hard cases have no polynomial size circuits for counting the number of homomorphisms but it was
unclear if these problems are \vnp complete. The study of \vnp complete problems and the arithmetic
world was started in the seminal paper by Valiant \cite{Val79b}. In this world, we look at the
complexity of computing a family of polynomials using a family of arithmetic circuits. Recently, a
dichotomy for graph homomorphisms was shown by
Rugy-Altherre \cite{springerlink:10.1007/978-3-642-32589-2_29}. Here a graph is encoded by a product of edge
variables and sets of graphs as sums over these products. This is known as generating function and a
detailed definition will be provided in \Cref{sec:general-stuff}. However, this result was for the
first side of the graph homomorphism problem.

In this paper we look at the second side of the graph homomorphism problem to complete the picture
for the arithmetic circuit world. While we could not get a general theorem as in
\cite{dalmau2004complexity}, we show multiple hardness proofs for some classes. We will look at
cycles, cliques, trees, outerplanar graphs, planar graphs and graphs of bounded genus.

Recently, homomorphism polynomials in a different form are even used for giving natural characterizations of \vp
independent of the circuit definition \cite{MahajanFSTTCS}. In this way our results can be interpreted as showing that
some straightforward candidates originating from the counting world do not give a characterization of \vp.

\Cref{sec:general-stuff} gives a formal introduction to our model, related hard problems and states the problem
precisely. We prove our dichotomies in
\Cref{sec:cycles,sec:cliques,sec:trees,sec:outerplanar-graphs,sec:planar-graphs,sec:genus-k-graphs}
where the constructions in \Cref{sec:outerplanar-graphs,sec:planar-graphs,sec:genus-k-graphs} build
on each other. The construction in \Cref{sec:trees} will use a slightly different model as the other
sections. We will give a brief introduction into concepts from graph genus in
\Cref{sec:genus-k-graphs} but refer the reader to the textbook by Diestel \cite{diestel2000graph}.

\section{Model and Definitions}
\label{sec:general-stuff}
Let us first give a brief introduction to the field of Valiant's classes. For further information
the reader is referred to the textbook by B{\"u}rgisser \cite{burgisser2000completeness}. In this theory, we
are given an arithmetic circuit (a directed acyclic connected graph) with addition and
multiplication gates over some field $\field$. These gates are either connected to other gates or
input gates from the set $\field \cup X$ for some set of indeterminates $X$. At the end we have
exactly one output gate. An arithmetic circuit computes a polynomial in $\field[X]$ at the output
gate in the obvious way.

As Valiant's model is non-uniform, a problem consists of families of polynomials.  A $p$-family is a
sequence of polynomials $(f_{n})$ over $\field[X]$ where the number of variables is $n$ and the
degree is bounded by some polynomial in $n$. Additionally the family of polynomials $(f_{n})$ should
be computed by a family of arithmetic circuits $(C_{n})$ where $f_{n}$ is computed by $C_{n}$ for
all $n$. Valiant's Model focuses its study on $p$-families of polynomials.

We define $L(f)$ to be the number of gates for a minimal arithmetic circuit computing a given
polynomial $f\in \field[X]$. \vp is the class of all $p$-families of polynomials where $L(f_{n})$ is
bounded polynomially in $n$. Let $q(n),r(n),s(n)$ be polynomially bounded functions. A $p$-family
$(f_{n})\in \field[x_{1},\dots,x_{q(n)}]$ is in \vnp if there exists a family
$(g_{n})\in \field[x_{1},\dots,x_{r(n)},y_{1},\dots,y_{s(n)}]$ in \vp such that
\[
  f(x_{1},\dots,x_{q(n)})=\sum_{\epsilon \in \{0,1\}^{s(n)}}
  g(x_{1},\dots,x_{r(n)},\epsilon_{1},\dots,\epsilon_{s(n)}).
\]
The classes \vp and \vnp are considered algebraic analogues to P and NP or more accurately
$\sharpp$. We can also define an algebraic version of AC$_{0}$, mentioned by
Mahajan and Rao \cite{DBLP:journals/cc/MahajanR13}. A $p$-family is in \vac if there exists a family of arithmetic
circuit of constant depth and polynomial size with unbounded fan-in that computes the family of
polynomials.

The notion of a reduction in Valiant's model is given by $p$-projections. A $p$-family $(f_{n})$ is
a $p$-projection of $(g_{n})$, written as $(f_{n})\leq_{p} (g_{n})$, if there exists a polynomially
bounded function $q(n)$ such that for every $n$, $f(x_{1},\dots,x_{n}) = g(a_{1},\dots,a_{q(n)})$ for
some $a_{i}\in \field \cup \{x_{1},\dots,x_{n}\}$. Once we have a reduction, we get a notion of
completeness in the usual way.

However, we use a different kind of reduction called a $c$-reduction. This is similar to a Turing
reduction in the Boolean world. We define $L^{g}(f)$ as the number of gates for computing $f$ where
the arithmetic circuits is enhanced with an oracle gate for $g$. An oracle gate for the polynomial
$g\in \field[x_{1},\dots,x_{n'}]$ has as output $g(a_{1},\dots,a_{n'})$ where $a_{1},\dots,a_{n'}$
are the inputs to this gate. This allows us to evaluate $g$ on $a_{1},\dots,a_{n'}$ in one step if we
computed $a_{1},\dots,a_{n'}$ previously in our circuit.

We say $f$ $c$-reduces to $g$, written $(f_{n})\leq_{c} (g_{n})$, if there exists a polynomial $p$
such that $L^{g_{p(n)}}(f)$ is bounded by some polynomial. This reduction, however, is only useful
for \vnp and not for \vac and \vp. In this paper we will exclusively deal with $c$-reductions for
our \vnp completeness results.

\subsection{Complete Problems}
\label{sec:complete-problems}
We continue with the basic framework of graph properties. In the following $\field$ will be an
infinite field.

\begin{definition}
    Let $X$ be a set of indeterminates. Let $\mathcal{E}$ be a graph property, that is, a class of
    graphs which contains with every graph also all of its isomorphic copies. Let $G=(V,E)$ be an edge
    weighted, undirected graph with a weight function $w:E\rightarrow \field \cup X$. We extend the
    weight function by $w(E'):=\prod_{e\in E'} w(e)$ to subsets $E'\subseteq E$.

    The \emph{generating function} $\gf(G,\mathcal{E})$ of the property $\mathcal{E}$ is defined as
    \[
    \gf(G,\mathcal{E}) := \sum_{E'\subseteq E} w(E')
    \]
    where the sum is over all subsets $E'$ such that the subgraph $(V,E')$ of $G$ has property
    $\mathcal{E}$.
\end{definition}
The reader should notice that the subgraph still contains all vertices and just takes a subset of
the edges.

In the following, let $G$ be a graph and let $X= \{x_{e} \ |\ e\in E\}$. We label each edge $e$ by
the indeterminate $x_{e}$.  We conclude by stating some basic \vnp-complete problems. Proofs of these
facts can be found in the textbook by B{\"u}rgisser \cite{burgisser2000completeness}.
\begin{theorem}[\cite{burgisser2000completeness}]\label{thm:burguhc}
    $\gf(K_{n}, \mathcal{UHC}_{n})$ is \vnp-complete where $\mathcal{UHC}_{n}$ is the set of all hamiltonian
    cycles in $K_{n}$.
\end{theorem}

\begin{theorem}[\cite{burgisser2000completeness}]\label{thm:burgclique}
    Let $\mathcal{CL}$ be the set of all cliques. Meaning, the set of all graphs, where one
    connected component is a complete graph and each of the remaining connected components consist
    of one vertex only. The family $\gf(K_{n},\mathcal{CL})$ is \vnp-complete.
\end{theorem}

\begin{theorem}[\cite{burgisser2000completeness}]
    Let $\mathcal{M}$ be the set of all graphs where all connected components have exactly two
    vertices. The family $\gf(K_{n},\mathcal{M})$ is \vnp-complete.
\end{theorem}
This polynomial gives us all perfect matchings in a graph. It is well known that the original
\vnp-complete problem, the permanent, is equal to $\gf(K_{n,n},\mathcal{M})$ for bipartite graphs
which is a projection of $\gf(K_{n^{2}},\mathcal{M})$.

\subsection{The problem and related definitions}
\label{sec:problem}
We now formulate our problem. Let $G,H$ be undirected graphs. We will generally switch freely
between having the variable indexed by either edges ($x_e$) or vertices ($x_{i,j}$ for $i,j\in V$). We
let $x_{j}$ correspond to the self-loop at vertex $j$.

A homomorphism from $G=(V,E)$ to $H=(V',E')$ is a mapping $f:V\rightarrow V'$ such that for all edges
$\{u,v\}\in E$ there exist an edge $\{f(u),f(v)\}\in E'$. We can define the corresponding generating
function as follows.

\begin{definition}\label{prob:hom}
    Let $\mathcal{H}_{H}$ be the property of all connected graphs homomorphic to a fixed $H$. We
    denote by $\mathcal{F}^{H,n}$ the generating function
    $\mathcal{F}^{H,n} :=\gf(K_n,\mathcal{H}_{H}).$
\end{definition}

We can state now the first dichotomy theorem.
\begin{theorem}[\cite{springerlink:10.1007/978-3-642-32589-2_29}]\label{thm:althere}
    If $H$ has a loop or no edges, $\mathcal{F}^{H,n}$ is in \vac and otherwise it is \vnp-complete
    under $c$-reductions.
\end{theorem}

Instead of looking at all graphs, we want to look at a restricted version. What happens if we do not
want to find every graph homomorphic to a given $H$ but every \emph{cycle} homomorphic to a given
$H$? We state our problem in the next definitions.
\begin{definition}\label{prob:ourhom}
    Let $\mathcal{E}_{n}$ be a graph property. Then $\mathcal{F}^{H,n}_{\mathcal{E}_{n}}$ is the generating function for all graphs in
    $\mathcal{E}_{n}$ on $n$ vertices homomorphic to a fixed graph $H$.
\end{definition}

\begin{definition}
    We define the following graph polynomials.
    \begin{itemize}
        \item $\mathcal{F}^{H,n}_{\text{cycle}_{n}}$ where $\text{cycle}_{n}$ is the property
        where one connected component is a cycle and the others are single vertices in a graph of
        size $n$.        
        \item $\mathcal{F}^{H,n}_{\text{clique}_{n}}$ where $\text{clique}_{n}$ is the
        property where one connected component is a clique and the others are single vertices in a
        graph of size $n$.
        \item $\mathcal{F}^{H,n}_{\text{tress}_{n}}$ where $\text{trees}_{n}$ is the
        property where one connected component is a tree and the others are single vertices in a
        graph of size $n$.
        \item $\mathcal{F}^{H,n}_{\text{outerplanar}_{n}}$ where $\text{outerplanar}_{n}$ is the
        property where one connected component is a outerplanar graph and the others are single vertices
        in a graph of size $n$.
        \item $\mathcal{F}^{H,n}_{\text{planar}_{n}}$ where $\text{planar}_{n}$ is the
        property where one connected component is a planar graph and the others are single vertices
        in a graph of size $n$.
        \item $\mathcal{F}^{H,n}_{\text{genus(k)},n}$ where $\text{genus(k),n}$ is the
        property where one connected component has genus $k$ and the others are single vertices in a
        graph of size $n$.
    \end{itemize}
\end{definition}

We will use the notation $\Fcycle$, $\Fclique$, $\Ftree$, $\Fouterplanar$, $\Fplanar$ and $\Fgenus$
as a shorthand.

Let us now introduce the homogeneous degree of a polynomial.
\begin{definition}\label{def:homcomp}
    Let $\bar{x}=x_{i_{1}},\dots,x_{i_{l}}$ be a subset of variables and $(f_{n})$ be a
    $p$-family. We can write $f_{n}$ as
    \[
    f_{n} =\sum_{\bar{i}} \alpha_{\bar{i}} \prod_{j=1}^{n} x_{j}^{i_{j}}.
    \]
    The homogeneous component of $f_{n}$ of degree $k$ with variables $\bar{x}$ is
    \[
    \homc_{k}^{\bar{x}}(f_{n}) = \sum_{\substack{i_{1},\dots,i_{l}\\
            k=\sum_{j=1}^{l} i_{j}}} \alpha_{i_{1},\dots,\i_{l}} x_{i_{1}}^{i_{i}}\dots x_{i_{l}}^{i_{j}}. 
    \]
\end{definition}

Finally, we need a last lemma in our proofs. This lemma was stated explicit by
Rugy-Altherre \cite{springerlink:10.1007/978-3-642-32589-2_29} and can also be found in
\cite{burgisser2000completeness}. It will give us a way to extract all polynomials of homogeneous
degree $k$ in some set of variables in $c$-reductions.
\begin{lemma}\label{lem:homdegree}
    Then for any sequence of integers $(k_{n})$ there exists a c-reduction from the homogeneous
    component to the polynomial itself:
    \[
    \homc_{k_{n}}^{\bar{x}}(f_{n}) \leqc (f_{n}).
    \]
    The circuit for the reduction has size in $\mathcal{O}(n\delta_{n})$ where $\delta_{n}$ is
    the degree of $f_{n}$.
\end{lemma}
The reader should note that using this theorem will blow up our circuit polynomially in size and can
hence be used only a constant number of times in succession.  However, we can use this lemma on
subsets of vertices. We replace every variable $x_{i}$ in the subset by $x_{i}y$ for a new variable
$y$ and take the homogeneous components of $y$. We will use this technique to \emph{enforce} edges
to be taken. Notice that enforcing $n$ edges to be taken only increases the circuit size by a factor
of $n$. Additionally, we can set edge variables to zero to \emph{deny} our polynomial using these edges.

Let $G$ be a graph that is homomorphic to a given $H$. We will, in general, ignore self-loops in
$G$, \ie assume $G$ to never have any self-loops. If we have proven a theorem for all $G$ without
self loops, we can just take the homomorphism polynomial with self-loops, take the homogeneous
component of degree zero of all self-loops and get the homomorphism polynomial without
self-loops. As we will prove the dichotomy for these, the hardness will follow. 

\section{Dichotomies}
\subsection{Cycles}
\label{sec:cycles}
As a first graph class we look at cycles. The proof for the dichotomy will be relatively easy and
gives us a nice example to get familiar with homomorphism polynomials and hardness proofs. Our proofs
will, in general, reason first about the kind of monomials that exist for a given $H$ and then try to
extract or modify these via \Cref{lem:homdegree} to get a solution to a \vnp-complete problem. This
will yield the reduction.

Our main dichotomy for cycles is the following theorem.
\begin{theorem}\label{thm:cyclesvnp}
    If $H$ has at least one edge or has a self-loop, then $\Fcycle$ is \vnp-complete under $c$-reductions. Else it is in
    \vac.
\end{theorem}

The next simple fact shows us which cycles are homomorphic to a given graph $H$. Let $n_{0}$ be
defined as $n$ if $n$ is even and $n-1$ if $n$ is odd.
\begin{fact}\label{fact:cyclehom}
    Given $H$ a graph with at least one edge, all cycles of length $n_{0}$ are homomorphic to $H$.
\end{fact}
It is easy to see that by folding the graph in half we get one path which is trivially homomorphic
to an edge. Our hardness proof will only be able to handle cycles of even length. Luckily this is
enough to prove hardness.
\begin{lemma}\label{lem:cycleeventransfer}
    Let $\cycleeven$ be the graph property of all cycles of length $n_{0}$.

    Then $\gf(K_{n_{0}},\cycleeven)$ is \vnp-hard under c-reductions.
\end{lemma}
\begin{proof}
    If $n$ is even, we can immediately use the hardness of $\gf(K_{n},\mathcal{UHC}_{n})$ (cf.\
    \Cref{thm:burguhc}). If $n$ is odd, we use the following reduction. We have given all cycles of
    length $n-1$ and want to get all cycles of length $n$. We evaluate the polynomial for $K_{n+1}$
    and get all cycles of length $n+1$. We can contract one edge with the following argument. We
    enforce, via taking the homogeneous component of degree one of $x_{n+1,1}$, all cycles to use
    $x_{n+1,1}$. We then replace $x_{i,n+1}$ by $x_{i,1}$ for all $i$ and set $x_{n+1,1}$ to
    one. This gives us all cycles of length $n$ with a factor $2$ for every monomial.

    To see this let us look at the following argument. Let the edge $(n+1,1)$ be the edge we
    contract and let $i,j$ be two arbitrary points picked in the graph. If we connect $i,j$ with a
    path through every point we can complete this into a cycle two different ways. Either with the
    edge $(1,i),(n+1,j)$ or $(1,j),(n+1,i)$. Notice, that every different choice of $i,j$ will
    construct a different cycle if we contract $1$ and $n+1$.

    This concludes our reduction to $\gf(K_{n},\mathcal{UHC}_{n})$. As our circuit can easily divide
    by two if the polynomial is over an infinite field (see. \cite{Strassen1973}).
\end{proof}
Later proofs will also use the contracting idea from the previous lemma.
A simple case distinction will give us the proof of the theorem.
\begin{proof}[Proof of \Cref{thm:cyclesvnp}]

    If $H$ has at least one edge, we know from \Cref{fact:cyclehom} that all even cycles are
    homomorphic to $H$ and by this represented in our polynomial. If we take the homogeneous
    components of degree $n_{0}$, we extract all even cycles of length $n_{0}$. This is \vnp-hard
    via the previous Lemma (\ref{lem:cycleeventransfer}).
    
    If $H$ has a self-loop, we can map all cycles to the one vertex in
    $H$. We can then extract the hamiltonian cycles of length $n$ by using the homogeneous degree of
    $n$ as all cycles are homogeneous to a self-loop.

    If $H$ has no edge, our polynomial is the zero polynomial as we cannot map any graph $G$
    containing an edge to $H$.

    Using Valiant's Criterion, we can prove membership of $\Fcycle$ in \vnp
    (cf.\cite{burgisser2000completeness}).
\end{proof}

\subsection{Cliques}
\label{sec:cliques}
Here, we will not use cycles in the hardness proof but work directly with the clique polynomial
defined by Bürgisser. The complete proof is an easy exercise. In contrast to the other results, we
show that computing $\Fclique$ is easy for most choices of $H$.

\begin{theorem}\label{thm:vnpclique}
    If $H$ has a self-loop then $\Fclique$ is \vnp-complete under c-reductions. Otherwise $\Fclique$ is in \vac.
\end{theorem}
\begin{proof}
    Let $H$ have at least one edge and no self-loop. We can use that $H$ has constant size which
    implies that $H$ has a maximal subgraph which forms a clique or $H$ has no clique. If $H$ has no
    clique, only a single edge or a single vertex is homomorphic to $H$.

    Let us now look at the case for cliques of size $c$. We can compute $\Fclique$ explicit by a
    brute-force algorithm. The number of monomials can be bound by the following argument. There are
    $\sum_{i=2}^{c}\binom n i$ many different cliques. As we can bound $\binom n i$ by $n^{i}$ we
    get an upper bound of $cn^{c}$ monomials. Further inspection yields, that constant depth,
    unbounded fan-in circuits of polynomial size are enough to compute all cliques up to size
    $c$.
    
    Let $H$ now have one self-loop. The fact that all cliques are homomorphic to a given graph with
    a self-loop tells us that $\Fclique$ contains different monomials for all cliques of size $i$
    for $i=1\dots,n$. The \vnp-hardness follows via \Cref{thm:burgclique}.

    The empty graph has the zero polynomial.

    As a polynomial time deterministic machine can easily check if a given instance is a clique, we can use
    Valiant's Criterion to show membership in \vnp.
\end{proof}

\subsection{Trees}
\label{sec:trees}
As the new characterization of VP had a specific tree structure we want to look at the general
problem. In previous sections our polynomial just contained the edges of the graph but for this
section we need a slightly different model. If a monomial in our polynomial would select the edges
$E'$ we also select the vertices $\{u,v | \{u,v\}\in E'\}$ in our monomial. In essence, we will also
select the vertices forming the edges, giving us polynomials with variables
$X=\{x_{e} | e\in E\} \cup \{x_{v} | v\in V \}$. It will be clear later why we need this special
form.

\begin{theorem}\label{thm:treehard}
    If $H$ contains an edge, then $\Ftree$ is \vnp-complete under c-reductions. Otherwise $\Ftree$ is in \vac.
\end{theorem}
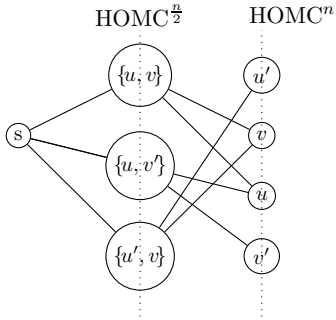
\begin{figure}
    \centering
    \begin{tikzpicture}[scale=0.8, every node/.style={transform shape}]
        \node[nu] (s) at (0,0) {s};
        \node[nu] (uv) at (2,1) {$\{\!u,v\!\}$};
        \node[nu] (uvp) at (2,-0.5) {$\{\!u,v'\!\}$};
        \node[nu] (upv) at (2,-2) {$\{\!u',v\!\}$};

        \node[nu] (u) at (4,-1) {$u$};
        \node[nu] (v) at (4,0) {$v$};
        \node[nu] (up) at (4,1) {$u'$};
        \node[nu] (vp) at (4,-2) {$v'$};

        \draw (s) -- (uv);
        \draw (s) -- (uvp);
        \draw (s) -- (uvp);
        \draw (s) -- (upv);
        \draw (uv) -- (u);
        \draw (uv) -- (v);
        \draw (uvp) -- (u);
        \draw (uvp) -- (vp);
        \draw (upv) -- (up);
        \draw (upv) -- (v);

        \node at (4.5,2) {$\homc^{n}$};
        \draw[dotted] (4,-3) -- (4,1.75);
        \node at (2,2) {$\homc^{\frac n 2}$};
        \draw[dotted] (2,-3) -- (2,1.75);
    \end{tikzpicture}
    \caption{Reduction from Trees to matching\label{fig:treematch}}
\end{figure}
\begin{proof}
    We use a reduction from connected partial trees to perfect matchings. It is obvious that a tree
    is always homomorphic to one edge.

    We want to compute a matching on a graph given by $(V,E)$. We can build a graph as in
    \Cref{fig:treematch} from a $K_{n}$ by setting the weight of every edge not given to zero. In
    detail, our graph has vertices $\{v\in V\} \cup \{v_e | e\in E\} \cup \{s\}$. We add the edges
    $\{(u,v),u\}, \{(u,v),v\}$ and $\{s,v_e\}$ for every $e\in E$. Vertices of the form
    $\{v_e | e\in E\}$ will be called \emph{edge-vertices} in this proof. Now as the vertices are
    given by our polynomials we can take the homogeneous components over vertices. We take the
    homogeneous components of degree $n/2$ over vertices $\{v_e| e\in E\}$ and of degree $n$ of
    vertices $v\in V$. Our matching in the original graph is given by the edges $(s,v_e)$.
    
    Every matching in the original graph has obviously a tree in our graph. Left to prove is the
    other direction. Given a tree in our graph, we know that only $n/2$ edge-vertices are
    selected. As every vertex $v\in V$ has to be connected by an edge, edge-vertices have to go to
    pairwise different sets of $v\in V$. Hence we can compute a perfect matching which is as hard as
    computing the permanent.
    
    Valiant's Criterion will again show the membership.
\end{proof}
We crucially need the fact that we get the adjacent vertices for free in our homomorphism
polynomials. The reader might think restricting the edges out of $s$ might suffice but this is not
the case. We could have a path that starts from $s$ goes over an edge-vertex to a vertex $u$ takes
the edge from $u$ to some other edge-vertex and continues until we have connected all edge-vertices
and all vertices into a path. This is obviously not a matching. If we want to forbid this behaviour,
we might want to select all edges outgoing from $s$. This would prevent the above case but the
reconstruction of a matching is non trivial.

An interesting fact of the proof is, that it does not use the fact that the graph class only
contains trees. Instead we only use that it contains trees. Hence the theorem can be easily extended
to other graph classes, provided we look at the homomorphism polynomials which contain edges and the
vertices connected to these edges.
\begin{corollary}
    Let $\mathcal{C}$ be a graph class containing all trees of size $n$. Then the following theorem holds on
    homomorphism polynomials containing edge and vertex variables from $\mathcal{C}$ to a given $H$. If $H$
    contains an edge, then the homomorphism polynomial is \vnp-complete under
    c-reductions. Otherwise it is in \vac.
\end{corollary}
\begin{proof}
    We can easily set the weight of every vertex not in our tree to zero and construct the same
    reduction as in the theorem.
\end{proof}

\subsection{Outerplanar Graphs}
\label{sec:outerplanar-graphs}
Next we will show a dichotomy for outerplanar graphs. We start with the case of a triangle
homomorphic to $H$.
\begin{figure}
    \centering
    \begin{subfigure}[b]{0.32\textwidth}
        \centering
        \begin{tikzpicture}[scale=0.75, every node/.style={transform shape}]
            \node (A) [draw,regular polygon, regular polygon sides=12, minimum size=2cm,outer sep=0pt] {};
            \foreach \n in {1,...,12} {
                \node at (A.corner \n) [anchor=360/5*(\n-1)+270] {};
            }
            \draw[very thick] (A.corner 1) -- (A.corner 3);
            \draw[ultra thick] (A.corner 1) -- (A.corner 4);
            \draw[ultra thick] (A.corner 1) -- (A.corner 5);
            \draw[ultra thick] (A.corner 1) -- (A.corner 6);
            \draw[ultra thick] (A.corner 1) -- (A.corner 7);
            \draw[ultra thick] (A.corner 1) -- (A.corner 8);
            \draw[ultra thick] (A.corner 1) -- (A.corner 9);
            \draw[ultra thick] (A.corner 1) -- (A.corner 10);
            \draw[ultra thick] (A.corner 1) -- (A.corner 11);

            \draw[ultra thick] (A.corner 1) -- (A.corner 2);
            \draw[ultra thick] (A.corner 1) -- (A.corner 12);
        \end{tikzpicture}
        \caption{Triangle graph\label{fig:optriangle}}
    \end{subfigure}
    \begin{subfigure}[b]{0.32\textwidth}
        \centering
        \begin{tikzpicture}[scale=0.75, every node/.style={transform shape}]
            \node (A) [draw,regular polygon, regular polygon sides=13, minimum size=2cm,outer sep=0pt] {};
            \foreach \n in {1,...,13} {
                \node at (A.corner \n) [anchor=360/7*(\n-1)+270] {};
            }
            \draw[ultra thick] (A.corner 1) -- (A.corner 2);
            \draw[ultra thick] (A.corner 1) -- (A.corner 4);
            \draw[ultra thick] (A.corner 1) -- (A.corner 6);
            \draw[ultra thick] (A.corner 1) -- (A.corner 8);
            \draw[ultra thick] (A.corner 1) -- (A.corner 10);
            \draw[ultra thick] (A.corner 1) -- (A.corner 12);

            \draw[ultra thick] (A.corner 2) -- (A.corner 3);
            \draw[ultra thick] (A.corner 4) -- (A.corner 5);
            \draw[ultra thick] (A.corner 6) -- (A.corner 7);
            \draw[ultra thick] (A.corner 8) -- (A.corner 9);
            \draw[ultra thick] (A.corner 10) -- (A.corner 11);
            \draw[ultra thick] (A.corner 12) -- (A.corner 13);
        \end{tikzpicture}
        \caption{Illustration of graph with buddy vertices\label{fig:buddy}}
    \end{subfigure}
    \begin{subfigure}[b]{0.32\textwidth}
        \centering
        \begin{tikzpicture}[scale=0.75, every node/.style={transform shape}]
            \node[nu] (m) at (0,0) {$c$};
            \node[nu] (j) at (-1,-1) {$p(v)$};
            \node[nu] (jp) at (0,-1) {$u$};
            \node[nu] (jpp) at (1,-1) {$u'$};
            \node[nu] (i) at (0.75,-2) {$v$};
            \draw[ultra thick] (m) -- (j);
            \draw[ultra thick] (m) -- (jp);
            \draw[ultra thick] (m) -- (jpp);
            \draw[ultra thick] (m) -- (i);

            \draw (i) -- (j);
            \draw (i) -- (jp);
            \draw (i) -- (jpp);
        \end{tikzpicture}
        \caption{\label{fig:threeedgeouterplanar}}
    \end{subfigure}    
    \caption{}
\end{figure}
\begin{lemma}\label{thm:optrianglehard}
    If a triangle is homomorphic to $H$ then $\Fouterplanar$ is \vnp hard under c-reductions.
\end{lemma}
\begin{proof}
    We will reduce to Hamiltonian Cycle by using a construction as in \Cref{fig:optriangle}. This
    means, we pick an arbitrary vertex $c$ and enforce all $n$ outgoing edges from this vertex
    via homogeneous components. We further enforce the whole graph to have $n+n-3$ edges. The graph
    given is obviously outerplanar but we still need to proof that no other graph fulfilling our
    criteria can be outerplanar.

    We call the implied order of the graph, the order of the outer circle of vertices starting from
    the star and ending at it again without any edges crossing. As there are two such orderings let
    us fix an arbitrary one for every graph. Let us now look at a graph which has not an implied
    order of the outer vertices. This implies that there exists a vertex $u$ which has degree
    4. With our ordering every vertex (except $c$ up to and including the later defined vertex $v$
    has a single parent. Furthermore, let $v$ be the first vertex of degree 4 in this order and let
    $p(v)$ be the parent of $v$. Notice that by enforcing all $n$ instead of just $n-2$ edges
    starting at the center, a parent $p(v)\neq c$ has to exist.

    Let $u,u'$ denote the other vertices adjacent to $v$ different than $p(v)$ and $c$. As we
    enforced edges from $c$ to every vertex, we can easily see the $K_{2,3}$ with $v,c$ on the one
    side and $u,u',p(v)$ on the other side. Hence the graph cannot be outerplanar. This implies that
    every vertex except $c$ and the two neighbouring vertices have degree at most
    3. Enforcing the overall number of edges gives us at least degree 3 and hence implies equality.

    From this we can reconstruct all cycles in a $K_{n-2}$. We need to remove the center of the star
    and glue the two vertices on the cycle next to the center together. We do this by a similar
    argument as in the proof for \Cref{lem:cycleeventransfer}. We evaluate the other enforced edges
    with one to get all cycles in a $K_{n-2}$ where every monomial is weighted by 2. Division again
    gives us the correct polynomial. Taking the homogeneous components as described only increases
    the circuit by a factor of $n$.
\end{proof}

\begin{theorem}
    If $H$ has an edge then $\Fouterplanar$ is \vnp-complete under c-reductions and otherwise trivial.
\end{theorem}
\begin{proof}
    To make the graph homomorphic to a single edge we will modify it in the following way. For every
    vertex $v$, except $c$, we choose a \emph{buddy vertex} $v'$. We enforce the edge between
    every vertex and his buddy vertex and set the edge between a buddy vertex and $c$ to
    zero. Additionally, we set all vertices from $v$ to any other non buddy vertex to zero and all
    edges from a buddy vertex to a different buddy vertex to be zero. In essence this splits every
    vertex into a left and right part (see \Cref{fig:buddy}). The hardness proof follows from
    \Cref{thm:optrianglehard} by contracting the edge between a vertex and his buddy vertex. Hence
    the combined degree of a vertex and his buddy vertex is at most three. Taking the homogeneous
    components increases the circuit size by a factor of $n$.
    
    We know by~\cite{Mitchell1979229} that checking if a graph is outerplanar is possible in linear
    time. With this we can use Valiant's Criterion to show the membership.
\end{proof}

\subsection{Planar Graphs}
\label{sec:planar-graphs}

\begin{figure}
    \centering
    \begin{tikzpicture}[scale=0.5]
        \node[k] (1) at (-3,0) {};
        \node[k] (2) at (-2,0) {};
        \node[k] (3) at (-1,0) {};
        \node[k] (4) at (0,0) {};
        \node[k] (5) at (1,0) {};
        \node[k] (6) at (2,0) {};
        \node[k] (7) at (3,0) {};

        \node[nu] (a) at (0,3) {a};
        \node[nu] (b) at (0,-3) {b};

        \draw[ultra thick] (1) -- (2);
        \draw (2) -- (3) -- (4) -- (5) -- (6);
        \draw[ultra thick] (6) -- (7);
        \draw[ultra thick] (a) -- (1) -- (b);
        \draw[ultra thick] (a) -- (2) -- (b);
        \draw[ultra thick] (a) -- (3) -- (b);
        \draw[ultra thick] (a) -- (4) -- (b);
        \draw[ultra thick] (a) -- (5) -- (b);
        \draw[ultra thick] (a) -- (6) -- (b);
        \draw[ultra thick] (a) -- (7) -- (b);
    \end{tikzpicture}
    \caption{Planar Gadget\label{fig:planar}}
\end{figure}
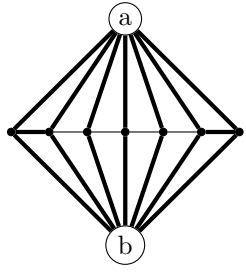

\begin{lemma}\label{lem:planar}
    All graphs isomorphic to \Cref{fig:planar}  with the thick edges fixed and $n+2+2(n+2)$ edges
    required are all permutations of the vertices $(1,\dots,n)$.
\end{lemma}
\begin{proof}
    Take an embedding in the plane of the graph without any crossings. If we show that every vertex
    has at most one edge going to the right, it follows that the set of vertices from left to right
    ordered is a permutation of the vertices.

    Let us look at the following subgraph. Let $v$ be a vertex with two right successors $u,u'$ and
    a parent $p(v)$. By construction the parent always exists. We denote the top and bottom vertex
    by $a$ and $b$ in our graph. We can now build a $K_{3,3}$ minor in the following
    way. $S_{1}=\{v,a,b\}$ and $S_{2}=\{u,u',p'\}$. As $a$ and $b$ are connected to every vertex we only
    need to check that $u$ is connected to $u,u'$ and $p$ which is by assumption. This proves that via
    edge deletion our graph would have a $K_{3,3}$ minor if the vertices would not give us a
    permutation.
\end{proof}

\begin{theorem}
    If $H$ has an edge then $\Fplanar$ is \vnp-complete under c-reductions. Otherwise $\Fplanar$ is in \vac.
\end{theorem}
\begin{proof}
    We again glue the second and second to last vertex in our planar gadget together in a similar
    manner as in the previous constructions to get all cycles from a path. Notice, how these are
    independent of the order and hence the same for all possible ordering.
    
    However, this graph is not yet homomorphic to a single edge. To accomplish this, we will use a
    graph of size $2n$. We, as in the outerplanar case, enforce every vertex, except $a$ and $b$, to
    have a buddy vertex $u_{v}$. Then we subdivide the edge $(a,v)$ and $(b,v)$ for every original,
    meaning none buddy, vertex $v$ with a new vertex $v'_{a}$, $v'_{b}$ respectively. This will give
    us for every part a square consisting of the vertices $a,v,v'_{a},u_{v}$ and the square
    $b,v,v'_{b},u_{v}$.

    Now it is easy to see that we can fold $a$ to $b$ which leaves us with a grid of
    height one. A grid can be easily folded to one edge. The size of the circuit is increased by a
    factor of at most $2n$.

    As testing planarity is easy, we can use Valiant's Criterion to show membership.
\end{proof}

\subsection{Genus $k$ graphs}
\label{sec:genus-k-graphs}
Graph embeddings are one of the major relaxations of planarity. For this we find a surface of a
specific type such that a graph can be embedded in this surface without any crossing edges. If we
want to increase the \emph{orientable genus} of a surface by one, we can glue a handle onto it which
edges can use without crossing other edges. We call a graph a \emph{genus $k$ graph} if there exists
a surface of orientable genus $k$ such that $G$ can be embedded in this surface and $k$ is
minimal. Notice, that a genus $0$ graph is planar. While the topic of graph genus is vast, we will
mostly use theorems as a blackbox and only reason about graphs of genus zero and one. For a detailed
coverage of the topic, the reader is referred to~\cite{diestel2000graph}.

With the planar result in place we can use the simple proof strategy. Construct a genus $k$ graph
where we append the planar construction. In this way the genus bound will ensure that our planar
gadget gives us all permutation of vertices as long as the connection of these two graphs will not
reduce the genus.
\begin{figure}
    \begin{subfigure}[b]{0.2\textwidth}
        \begin{tikzpicture}
            \node[nu] (1) at (0,0) {1};
            \node[nu] (2) at (0,1) {2};
            \node[nu] (3) at (1,1) {3};
            \node[nu] (4) at (1,0) {4};
            \node[nu] (5) at (-0.5,-0.5) {5};
            \node[nu] (6) at (-0.5,1.5) {6};
            \node[nu] (7) at (1.5,1.5) {7};
            \node[nu] (8) at (1.5,-0.5) {8};
            
            \draw (1) -- (2) -- (3) -- (4) -- (1);
            \draw (5) -- (6) -- (7) -- (8) -- (5);
            \draw (1) -- (3);
            \draw (2) -- (4);
            \draw (5) -- (1);
            \draw (6) -- (2);
            \draw (7) -- (3);
            \draw (8) -- (4);
        \end{tikzpicture}
        \caption{Gadget\label{fig:oneblock}}
    \end{subfigure}
    \begin{subfigure}[b]{0.35\textwidth}
        \begin{tikzpicture}[scale=0.9]
            \node[k] (1) at (0,0) {};
            \node[k] (2) at (0,1) {};
            \node[k] (3) at (1,1) {};
            \node[k] (4) at (1,0) {};
            \node[k] (5) at (-0.5,-0.5) {};
            \node[k] (6) at (-0.5,1.5) {};
            \node[k] (7) at (1.5,1.5) {};
            \node[k] (8) at (1.75,-0.5) {};

            \draw (1) -- (2) -- (3) -- (4) -- (1);
            \draw (5) -- (6) -- (7) -- (8) -- (5);
            \draw (1) -- (3);
            \draw (2) -- (4);
            \draw (5) -- (1);
            \draw (6) -- (2);
            \draw (7) -- (3);
            \draw (8) -- (4);
            
            \node[k] (21) at (2.5,0) {};
            \node[k] (22) at (2.5,1) {};
            \node[k] (23) at (3.5,1) {};
            \node[k] (24) at (3.5,0) {};
            \node[k] (26) at (2,1.5) {};
            \node[k] (27) at (4,1.5) {};
            \node[k] (28) at (4,-0.5) {};
            
            \draw (21) -- (22) -- (23) -- (24) -- (21);
            \draw (8) -- (26) -- (27) -- (28) -- (8);
            \draw (21) -- (23);
            \draw (22) -- (24);
            \draw (8) -- (21);
            \draw (26) -- (22);
            \draw (27) -- (23);
            \draw (28) -- (24);
        \end{tikzpicture}
        \caption{Two Gadgets\label{fig:twoblock}}
    \end{subfigure}
    \begin{subfigure}[b]{0.4\textwidth}
        \begin{tikzpicture}[scale=0.9]
            \node[k] (1) at (0,-0.25) {};
            \node[k] (2) at (0,0.75) {};
            \node[k] (3) at (1,0.75) {};
            \node[k] (4) at (1,-0.25) {};
            \node[k] (6) at (-0.5,1) {};
            \node[k] (7) at (1.5,1) {};
            \node[k] (8) at (1.5,-1) {};

            \node[k] (b1) at (-3,0) {};
            \node[k] (b2) at (-2.5,0) {};
            \node[k] (b3) at (-2,0) {};
            \node[k] (b4) at (-1.5,0) {};
            \node[k] (b5) at (-1,0) {};
            
            \node[k] (ba) at (-2,1) {};
            \node[k] (bb) at (-2,-1) {};

            \draw (1) -- (2) -- (3) -- (4) -- (1);
            \draw (b5) -- (6) -- (7) -- (8) -- (b5);
            \draw (1) -- (3);
            \draw (2) -- (4);
            \draw (b5) -- (1);
            \draw (6) -- (2);
            \draw (7) -- (3);
            \draw (8) -- (4);

            \draw (ba) -- (b1) -- (bb) -- (b2) -- (ba) -- (b3) -- (bb) -- (b4) -- (ba) -- (b5) --
            (bb);
            \draw (b1) -- (b2) --  (b3) -- (b4) -- (b5);
        \end{tikzpicture}
        \caption{Gadget with planar gadget\label{fig:planarblock}}
    \end{subfigure}
    \caption{\label{fig:blocks}}
\end{figure}

\begin{lemma}
    The graph in \Cref{fig:oneblock} has genus one.
\end{lemma}
\begin{proof}
    We can use the given embedding with one handle for the crossing in the middle to show an upper
    bound of one.

    We again construct a $K_{3,3}$ with the sets $S_{1}=\{2,1,6'\}, S_{2}=\{3,4,7'\}$ where $6'$ is
    the vertex constructed from contracting the edge $(5,6)$ and $7'$ from the edge $(7,8)$. And
    hence the graph is not planar and has a lower bound for the genus of one.
\end{proof}

The next theorem shows how we can glue graphs together to increase the genus in a predictable way.
\begin{definition}[\cite{battle1962additivity}]
    G is a \emph{vertex amalgam} of $H_{1},H_{2}$ if $G$ is obtained from disjoint graphs $H_{1}$
    and $H_{2}$ where we identify one vertex form $H_{1}$ with one vertex from $H_{2}$.
\end{definition}

With this we restate a theorem from Miller \cite{battle1962additivity} to compute the genus of a given graph.
\begin{theorem}[\cite{battle1962additivity}]\label{thm:addgenus}
    Let $\gamma(G)$ be the orientable genus of a graph $G$. Let $G$ be constructed from vertex
    amalgams of graphs $G_{1},\dots,G_{n}$. Then $\gamma(G)=\sum_{i=1}^n \gamma(G_{i})$.
\end{theorem}
This now gives us immediately the result that a graph constructed as in \Cref{fig:twoblock} with $k$
gadgets has genus $k$.
\begin{theorem}
    If $H$ has an edge then $\Fgenus$ is \vnp-complete under c-reductions for any $k$. Otherwise $\Fgenus$ is in \vac.
\end{theorem}
\begin{proof}
    With \Cref{thm:addgenus}, \Cref{lem:planar} and the construction in \Cref{fig:blocks} we are
    almost done. Because we enforced a genus $k$ graph to occur all graphs that homomorphic to the
    planar gadget have genus zero and hence be planar.

    The only thing left to do is to modify our graphs such that they are homomorphic to
    an edge without violating the properties. It is clear that we can fold our genus one gadgets
    together. If we then subdivide the edge $(1,3)$ and $(2,4)$ (which keeps our block property) we
    can first fold 7 to 5 and 3 to 1. Folding then again 6 to 8 and 2 to 4 we get a square with two
    dangling edges. The dangling edges can be folded onto the square and the square is homomorphic
    to one edge. This construction increases the size of the circuit at most by a factor of
    $14k+2n$. As testing for a fixed genus is in NP, we can use Valiant's Criterion to show
    membership.
\end{proof}

\section{Conclusion}
\label{sec:conclusion}
We have shown many dichotomy results for different graph classes but some classes are still open.
We want to especially mention the case of our graph class being the class of trees. It is known that
we can use Kirchoff's Theorem to find all spanning trees of a given graph. This, however, does not
include monomials of total degree less than $n-1$ which our polynomials include. From the algebraic
view, the knowledge ends here. In the counting view, where we solve the task of counting all trees
in a graph, a bit more is known. Goldberg and Jerrum \cite{goldberg2000counting} showed that counting the number of
subtrees that are distinct up to isomorphism is $\sharpp$-complete. This, combined with our
dichotomy for trees including the vertices, gives us a strong indication that the similar problem is
\vnp-hard in the algebraic world.

A different expansion of these results would be the case of bounded treewidth. As mentioned earlier,
in the counting version the case of bounded treewidth is indeed the most general form and completely
characterizes the easy and hard instances of counting graph homomorphisms. Additionally, recent advancements
showed that graph homomorphisms of a specific type characterize \vp. Can homomorphism from graph
classes parameterized by treewidth, similar
to the counting case, be used for a complete characterization of \vp and \vnp?

An interesting research direction would be the case of disconnected graph properties.
Rugy-Altherre looked at the property that any graph is
homomorphic to a given graph $H$. This includes disconnected graphs with connected components larger
than one vertex. We instead only looked at restricted homomorphisms where one major connected
component exists. It is unclear to the author if our proofs could be adapted to this case.

\paragraph*{Acknowledgments}
I want to thank my doctoral advisor M.\ Bläser for his guidance. I additionally want to thank
R.\ Curticapean for many discussions on the counting versions on problems and B.\ V.\ Raghavendra
Rao for introducing me to this topic.

\bibliographystyle{abbrv}

\end{document}